\pdfoutput=1
\documentclass{article}[12]
\usepackage[margin=1in]{geometry} 
\usepackage{enumerate, amssymb, amsfonts, latexsym}
\usepackage{amscd, amsmath, amsthm}
\usepackage{graphicx}
\usepackage{subfig}
\usepackage{bm}
\usepackage[ruled]{algorithm2e}
\usepackage{subfig}
\usepackage{mathtools}
\usepackage{float}
\theoremstyle{plain}
\newtheorem{theorem}{Theorem}[section]
\theoremstyle{definition}
\newtheorem{definition}[theorem]{Definition}
\newtheorem{lemma}[theorem]{Lemma}
\newtheorem{prop}[theorem]{Proposition}

\DeclarePairedDelimiter{\ceil}{\lceil}{\rceil}
\usepackage{hyperref}
\usepackage{wrapfig}
\hypersetup{colorlinks,    citecolor=red,
    filecolor=black,
    linkcolor=black,
    urlcolor=red}

\begin{document}

\title{Iterative Hard Thresholding for Weighted Sparse Approximation}
\author{Jason Jo \\ 
jjo@math.utexas.edu \\
Mathematics Department at the University of Texas at Austin
\\ 2515 Speedway,
Austin, Texas 78712}
\date{\today}
\maketitle
\begin{abstract}
Recent work by Rauhut and Ward developed a notion of weighted sparsity and
a corresponding notion of Restricted Isometry Property for the space of weighted
sparse signals. Using these notions, we pose a best
weighted sparse approximation problem, i.e. we seek structured sparse solutions
to underdetermined systems of linear equations. Many computationally
efficient greedy algorithms have been developed to solve the problem of best
$s$-sparse approximation. The design of all of these algorithms employ a
similar template of exploiting the RIP and computing projections onto the space
of sparse vectors. We present an extension of the Iterative Hard Thresholding
(IHT) algorithm to solve the weighted sparse approximation problem. This IHT
extension employs a weighted analogue of the template employed by all greedy
sparse approximation algorithms. Theoretical
guarantees are presented and much of the original analysis remains unchanged
and extends quite naturally. However, not all the theoretical analysis extends.
To this end, we identify and discuss the barrier to extension. Much like IHT,
our IHT extension requires
computing a projection onto a non-convex space. However unlike IHT and other
greedy
methods which deal with the classical notion of sparsity, no simple method is
known for computing projections onto these weighted sparse spaces. Therefore we
employ a surrogate for the projection and present its empirical performance on power law distributed signals. 
\end{abstract}


\allowdisplaybreaks

\section{Introduction}
Compressed sensing algorithms attempt to solve underdetermined linear systems of
equations by seeking structured solutions, namely that the underlying signal is
either sparse or well approximated by a sparse signal \cite{CSTextbook}.
However, in practice
much more knowledge about a signal's support set is known beyond that of sparsity
or
compressibility. Empirically it has been shown that the spectral power of
natural images decays with frequency $f$ according to a power-law $1/f^p$ for $p
\approx 2$ \cite{powerlaw1, powerlaw2}. Likewise, the frequency of earthquakes
corresponding to their magnitudes as measured by Moment magnitude scale
empirically also exhibits a power law decay \cite{natesilver}. For these types
of highly structured signals, certain atoms in
the dictionary are more prevalent in the support set of a signal than other atoms. The traditional notion of
sparsity treats all atoms uniformly and thus is not ideally suited to utilize
this rich prior knowledge. 

To this end one can consider using weighted $\ell_1$ minimization to obtain
structured sparse solutions. Weighted
$\ell_1$ minimization can leverage prior knowledge of a signal's support to
undersample the signal, and avoid overfitting the data \cite{RW, pdf1, pdf2,
pdf3}. However, one drawback
that weighted $\ell_1$ minimization shares
with $\ell_1$ minimization is that traditional solution methods scale poorly
\cite{Ambuj}. 

While many computationally efficient approximation algorithms
have been developed for computing a best $s$-sparse approximation
\cite{CSTextbook} no such method has been developed for the weighted case. In
this article, we make the following
contributions:
\begin{enumerate}
 \item Using a generalized notion of weighted sparsity and a corresponding
notion of Restricted Isometry Property on weighted sparse signals developed in
\cite{RW}, we pose a weighted analogue of the best $s$-sparse approximation
problem.
\item An extension of the Iterative Hard Thresholding (IHT) algorithm \cite{IHT}
is presented to solve the weighted
sparse approximation problem. We emphasize how the same template
used to derive performance guarantees for all the greedy compressed sensing
algorithms carries over naturally. Indeed, performance guarantees are derived
and much of the theoretical analysis remains unchanged. However, not all
theoretical results extend and the barrier seems to be the nature of weighted
thresholding. We explore this extension barrier and present a detailed analysis of
which theoretical guarantees do not extend and how the barrier is responsible
for this obstruction. Under an additional hypothesis, the extension barrier is
rendered moot and we present some specialized theoretical guarantees. The
nature and proof of these guarantees are all directly motivated by \cite{RW}.
\item While both IHT and the IHT extension compute a projection onto a
non-convex
space, the projection that IHT requires can actually be efficiently computed
while the projection that our IHT extension requires does not seem to have an
efficient solution. To this end, we consider a tractable surrogate to
approximate this non-convex projection and we present its empirical
performance on power law distributed signals.
\end{enumerate}

The remainder of the article is organized as follows. In Section 2 we quote
much of the weighted sparsity concepts developed from \cite{RW} that will be
needed for the IHWT extension. In Section 3 the IHWT algorithm is presented and
theoretical performance guarantees are established. In Section 4 we present
various numerical results. 

\section{Weighted Sparsity}
In this section, all of the concepts and definitions are taken from \cite{RW}.
For unstructured sparse recovery problems, the sparsity of a signal $\bm{x} \in
\mathbb{C}^N$ is defined to be the cardinality of its support set, denoted as
$\|\bm{x}\|_0$. More generally, we have a dictionary of atoms $\{a_i\}_{i =
1}^N$
and for the unstructured case, each atom is given the weight $\omega_i = 1$ for
all $i = 1, \dots, N$. In this context, the sparsity of a signal can be viewed
as the sum of
the weights of the atoms in the support set. Following \cite{RW}, given a
dictionary
$\{a_i\}_{i = 1}^N$ and a corresponding set of weights $\{\omega_i\}_{i =
1}^N, \omega_i \geq 1$ for $i = 1,\dots, N$, we can define the
\textit{weighted} $\ell_0$ \textit{norm}:
$$
\|\bm{x}\|_{\omega, 0} = \sum_{j: x_j \neq 0} \omega_j^2.
$$
Observe that the weighted sparsity of a vector $\bm{x}$ is at least as large as the
unweighted sparsity of $\bm{x}$, i.e. $\|\bm{x}\|_{\omega, 0} \geq
\|\bm{x}\|_0$.  

For any subset $S \subset \mathbb{N}$, we may define the weighted cardinality
of $S$ via:
$$
\omega(S) := \sum_{j \in S} \omega_j^2.
$$
In general, we also have the weighted $\ell_p$ spaces with norm:
$$
\|\bm{x}\|_{\omega, p} = \sum_{j: x_j \neq 0} |x_j|^p \omega_j^{2-p}.
$$

Using this generalized notion of sparsity allows us to pose the
\textit{best $(\omega,s)$-sparse approximation problem:}
\begin{equation}\label{weightedsparse}
 \textrm{minimize } \frac{1}{2}\|\bm{A} \bm{x} -\bm{y}\|_2^2 \textrm{    subject
to    } 
\|\bm{x}\|_{\omega,
0} \leq s.
\end{equation}

Given this generalized notion of sparsity, \cite{RW} defines a generalized
notion
of a map $\bm{A}: \mathbb{C}^N \rightarrow \mathbb{C}^m$ being an isometry on
the space of weighted sparse vectors:
\begin{definition} (Weighted restricted isometry constants) For $\bm{A} \in
\mathbb{C}^{m \times N}$, weight parameter $\omega$ and $s \geq 1$, the weighted
restricted isometry constant $\delta_{\omega, s}$ associated to $\bm{A}$ is the
smallest number $\delta$ for which
$$
 (1-\delta) \|\bm{x}\|_2^2 \leq \|\bm{Ax}\|_2^2 \leq (1+\delta)\|\bm{x}\|_2^2
$$
holds for all $\bm{x} \in \mathbb{C}^N$ with $\|\bm{x}\|_{\omega, 0} \leq s$. We
say that a map $\bm{A}$ has the weighted restricted isometry property with
respect to the weights $\omega$ ($\omega$-RIP) if $\delta_{\omega, s}$ is small
for $s$ reasonably large compared to $m$. 
\end{definition}

Observe that for any positive number $s$, there exists a partition of $s$ with
distinct parts of maximal cardinality, i.e. an index set $\mathcal{I}$ with
$\omega$-weighted cardinality $s$ with the largest number of non-zero atoms. Let
$P_{\omega}(s)$ denote this maximal term
$$
P_{\omega}(s) := \max_{\omega(\mathcal{I}) \leq s} |\mathcal{I}|. 
$$
Clearly if $\bm{A}$ satisfies RIP of order $P_{\omega}(s)$, then $\bm{A}$ will
also
satisfy the $\omega$-RIP of order $s$. However the converse does not hold. Not
only do weighted sparse signals have a constraint on the cardinality of their
support sets, they can also have a constraint on the maximal atom which can be
present in their support sets. Take for example the weights defined by
$\omega(j) = \sqrt{j}$, $j = 1,\dots, N$. An $(\omega,s)$-sparse signal cannot
have any atom with index higher than $\ceil{s^{1/2}}$ supported. If $\bm{A}$
were to
satisfy the $\omega$-RIP of order $s$, then the $\omega$-RIP alone does not
guarantee that $\bm{A}$ preserves the geometry of heavy-tailed signals, no
matter how
sparse they may be in the unweighted sense. We conclude that the $\omega$-RIP is
in general a weaker isometry condition than the RIP and the primary
reason being the existence of heavy tailed signals.  

In \cite{RW} a heuristic justification is given that for weights $\omega$
satisfying
$\omega(j) = j^{\alpha/2}$, with high probability an $m \times N$ i.i.d.
subgaussian random matrix satisfies the $\omega$-RIP once
$$
m = \mathcal{O} \left(   \alpha^{   1/\alpha - 1  } s^{1/(\alpha+1)} \log(s)  
\right).
$$
Note that fewer measurements are required than in the unweighted case, which has
the lower bound of $m = \mathcal{O}(s \log(N/s))$ measurements. 

The following properties of RIP matrices carry over immediately to $\omega$-RIP
matrices.
\begin{lemma}\label{classicRIPbounds}
Let $\bm{I}$ denote the $N \times N$ identity matrix. Given a set of weights
$\omega
\in \mathbb{C}^N$, vectors $\bm{u}, \bm{v} \in \mathbb{C}^N$, $\bm{y} \in
\mathbb{C}^m$ and an
index set $S \subseteq [N]$,
\begin{alignat*}{3}
| \langle  \bm{u}, (\bm{I} - \bm{A^*A})\bm{v} \rangle| &\leq \delta_{\omega,t}
\|\bm{u}\|_2 \|\bm{v}\|_2  &
\indent if \ &   \| \textup{supp}(\bm{u}) \cup      
\textup{supp}(\bm{v}) \|_{\omega,0}
\leq t,
\\ \| ((\bm{I} - \bm{A^*A})\bm{v})_S\|_2 &\leq \delta_{\omega, t}\|\bm{v}\|_2  &
   if \ &  \| S \cup       \textup{supp}(\bm{v})\|_{\omega,0} \leq t,
\\ \| (\bm{A^*y})_S\|_2 &\leq \sqrt{1+\delta_{\omega,s}}\|\bm{y}\|_2  &   if \  &
\|S\|_{\omega,0} \leq s. 
\end{alignat*} 
\end{lemma}
\begin{proof}
The proofs follow immediately from their unweighted counterparts where one
employs the $\omega$-RIP instead of the RIP. See
\cite{CSTextbook} for full proofs. 
\end{proof}

\section{Iterative Hard Weighted Thresholding: In Theory}
\subsection{Intuition and Background} First we
revisit the unweighted case to build some intuition and use it to motivate
modifications for the weighted case. In compressed sensing, greedy algorithms
solve problems of the form:
$$\min f(\bm{x}) \textrm{ subject to } \bm{x} \in \mathcal{S}$$ 
where $\mathcal{S}$ is some structured space and $f(\bm{x})$ denotes a
loss
function which depends on $\bm{x}$ and the vector of linear samples $\bm{y} =
\bm{Ax}^* + \bm{e}$,
where $\bm{e}$ represents measurement noise and $\bm{A} \in \mathbb{C}^{m \times
N}$ is a
sampling
matrix. For the case of best $s$-sparse approximation
$f(\bm{x}) = \frac{1}{2}\|\bm{Ax}-\bm{y}\|_2^2$ and $\mathcal{S} = \{\bm{x}
\in \mathbb{C}^N:
\|\bm{x}\|_0 \leq
s\}$.

Iterative greedy algorithms such as IHT exploit the RIP in the
following manner. If the matrix $\bm{A}$ satisfies the RIP of order $s$, then
$\bm{A}^*\bm{A}$
is a good enough approximation to the identity matrix on the space of $s$-sparse
vectors so that applying $\bm{A}^*$ to the vector of noisy samples
approximately yields the true signal $\bm{x}^*$ up to some ``noise term"
$$\bm{A}^*\bm{y} = \bm{A}^*\bm{Ax^*} + \bm{A}^*\bm{e} \approx \bm{x}^* +
\bm{\xi}.$$ 
Iterative greedy algorithms produce a dense signal approximation $\bm{z}^n
\approx
\bm{x}^* + \bm{\xi}^n$ at each stage $n$ and they denoise this dense signal to
output an
$s$-sparse approximation. Roughly speaking, the denoising process of all of
these iterative methods involves a projection onto the space of sparse vectors.
The idea behind this
is that while the noise term $\bm{\xi}^n$ is dense, its energy is assumed to be
spread
throughout all of its coordinates and as a result does not heavily contaminate the
$s$-most significant coordinates of $\bm{x}^*$. Despite the fact that
$\mathcal{S}$
is non-convex, it is very simple to compute projections onto this space: all
that is required is sorting the entries of the signal by their magnitude and
picking the top $s$ entries. This
fact combined with the above RIP intuition explains why greedy methods such as
IHT can efficiently solve a non-convex problem.

Setting our initial approximation $\bm{x}^{0} = \bm{0}$, IHT is the iteration:
\begin{equation}\label{IHT}
 \bm{x}^{n+1} = H_s(\bm{x}^{n} + \bm{A}^*(\bm{y} - \bm{Ax}^{n})),
\end{equation}
where $H_s$ is the \textit{hard thresholding operator} and at each step $n+1$
it outputs the best $s$-sparse approximation to $\bm{x}^{n} + \bm{A}^*(\bm{y} -
\bm{Ax}^{n})$ by
projecting it onto $\mathcal{S}$. More specifically
\begin{equation}
H_s(\bm{x}) = \inf_{\|\bm{z}\|_0 \leq s} \|\bm{x}-\bm{z}\|_2. 
\end{equation}
Note that by plugging in for $\bm{y} = \bm{Ax}^* + \bm{e}$,
we obtain the following approximation:
$$
\bm{x}^{n} + \bm{A}^*(\bm{y} - \bm{Ax}^{n}) = \bm{A}^*\bm{A} \bm{x}^* + (\bm{I}
- \bm{A}^*\bm{A})\bm{x}^n + \bm{A}^* \bm{e}
\approx \bm{x}^* + \bm{\xi}^n.
$$
Denoising by applying the hard thresholding operator $H_s$ yields
$\bm{x}^{n+1}$, an
$s$-sparse
approximation to the true underlying signal $\bm{x}$. 

\subsection{Extension to the Weighted Case}Observe that one can equivalently
view IHT as a projected gradient descent algorithm with constant
step size equal to 1. Once IHT is viewed in this manner, the modification we
make
to extend IHT to solve \eqref{weightedsparse} is quite natural:
we still
perform a constant step size gradient descent step at each iterate, however
instead of projecting onto the space of $s$-sparse vectors, we project onto the
space of weighted sparse vectors
$\mathcal{S}_{\omega,s} = \{ \bm{x} : \|\bm{x}\|_{\omega, 0} \leq s\}$. This
algorithm will
be
referred to as Iterative Hard Weighted Thresholding (IHWT)  and it is given by
the
following iteration
\begin{equation}
\label{IHWT} \bm{x}^{n+1} = H_{\omega,s} (\bm{x}^{n} + \bm{A}^*(\bm{y} -
\bm{Ax}^n)), 
\end{equation}
where $H_{\omega,s}$ is the \textit{hard weighted thresholding operator} and it
computes projections onto the space of weighted sparse vectors
$\mathcal{S}_{\omega,s}$
\begin{equation}\label{weightedthresholdingoperator}
H_{\omega,s}(\bm{x}) = \inf_{\|\bm{z}\|_{\omega,0} \leq s} \|\bm{x}-\bm{z}\|_2. 
\end{equation}

Computing the projection $H_{\omega,s}(\bm{x})$ is not as straightforward as
computing $H_{s}(\bm{x})$. In particular, sorting the signal by the magnitude of
its entries and then thresholding does not produce the best $(\omega,s)$-sparse
approximation. To see why, consider the simple example where $N = 3, \omega =
[1, \sqrt{2}, \sqrt{3}]$ and take the signal $\bm{x} = [9, 9, 10]$. For $s = 3$,
by sorting and thresholding, we obtain the following weighted 3 sparse
approximation $\bm{x}^* = [0,0,10]$ and $\|\bm{x}-\bm{x}^*\|_2 = 9\sqrt{2}$.
However, the other 3 sparse approximation $\widehat{\bm{x}} = [9,9,0]$ is in
fact a more accurate weighted 3 sparse approximation as $\|\bm{x} -
\widehat{\bm{x}}\|_2 = 10 < 9\sqrt{2}$. 

Therefore unlike the unweighted case, computing the best weighted $s$-sparse
approximation consists of a combinatorial search. To illustrate the difficulty
of executing this search, consider the case of a weight parameter $\omega$ given
by $\omega(j) = \sqrt{j}$ for $j = 1,\dots, N$. In this case, computing all the
possible index sets of weighted cardinality $s$ is equivalent to computing all
the
partitions of $s$ consisting of unique parts. With the square root weight
parameter, Wolfram Mathematica \cite{Wolfram} computes that there are 444,794
possible subsets
of weighted sparsity $s = 100$ and it computes that there are
8,635,565,795,744,155,161,506 support sets of size $s = 1000$. 

Despite this intractability, in the next subsection we derive theoretical
guarantees for the IHWT algorithm
and in Section 4, we will explore the empirical performance of a surrogate to
approximate the projection onto $\mathcal{S}_{\omega,s}$.

\subsection{Performance Guarantees} 
Throughout this subsection, we will employ the following notation:
\begin{enumerate}
\item $\bm{x}^s = H_{\omega,s}(\bm{x})$ as defined by
\eqref{weightedthresholdingoperator},
\item $\bm{r}^n = \bm{x}^s - \bm{x}^n$,
\item $S = \textrm{supp}(\bm{x}^s)$,
\item $\bm{x}_{\overline{S}} = \bm{x} - \bm{x}^s$,
\item $S^n = \textrm{supp}(\bm{x}^n)$,
\item $T^{n} = S \cup S^n$,
\item $\bm{a}^{n+1} = \bm{x}^n + \bm{A}^*(\bm{y}-\bm{Ax}^n) = \bm{x}^n +
\bm{A}^*(\bm{Ax} + \bm{e} - \bm{Ax}^n) =  \bm{x}^n + \bm{A}^*(\bm{Ax}^s +
\bm{Ax}_{\overline{S}} + \bm{e} - \bm{Ax}^n)$.
\end{enumerate}

\subsubsection{Performance Guarantees: Convergence to a Neighborhood}
Here we derive performance guarantees which establish that IHWT will converge to
a neighborhood of the best $(\omega,s)$-sparse approximation with a linear
convergence rate. The size of the neighborhood is dependent on how well the true
signal $\bm{x}$ is approximated by $\bm{x}^s$. 

In \cite{IHT}, the following performance guarantee was established \footnote{We
use
different notation than that of the original authors Blumensath and Davies.}:
\begin{theorem}\label{IHTguarantee}
Let $\bm{y} = \bm{Ax} + \bm{e}$ denote a set of noisy observations where $\bm{x}$ is an
arbitrary vector. Let
$\bm{x}^s$ be an approximation to $\bm{x}$ with no more than $s$ non-zero
elements
for
which $\|\bm{x}-\bm{x}^s\|_2$ is minimal. If $\bm{A}$ has restricted isometry
property with
$\delta_{3s} < 1/\sqrt{32}$, then at iteration $n$, IHT as defined by
\eqref{IHT} will recover an
approximation $\bm{x}^{n}$ satisfying
\begin{equation}\label{IHTrate}
\|\bm{x} - \bm{x}^{n}\|_2  \leq 2^{-n}\|\bm{x}^s\|_2 + 6 \widetilde{\epsilon}_s.
\end{equation} 
where
\begin{equation}\label{IHTerrorterm}
\widetilde{\epsilon}_s = \|\bm{x}-\bm{x}^s\|_2 + \frac{1}{\sqrt{s}}
\|\bm{x}-\bm{x}^s\|_1 + \|\bm{e}\|_2.
\end{equation}
\end{theorem}
In other words, IHT guarantees a linear convergence rate up to the
\textit{unrecoverable energy} $\widetilde{\epsilon}_s$. This
$\widetilde{\epsilon}_s$ term is referred to as unrecoverable energy as it
contains the measurement noise and energy terms which measure how well a signal
$\bm{x}$ can be approximated by sparse signals. 

We reverse course and instead focus our attention on an intermediate, yet more
general error bound for IHT:
\begin{theorem}\label{IHTGeneralGuaranteeTHM}
Let $\bm{y} = \bm{Ax} + \bm{e}$ denote a set of noisy observations where $\bm{x}$ is an
arbitrary vector. Let
$\bm{x}^s$ be an approximation to $s$ with no more than $s$ non-zero elements
for
which $\|\bm{x}-\bm{x}^s\|_2$ is minimal. If $\bm{A}$ has restricted isometry
property with
$\delta_{3s} < 1/\sqrt{32}$, then at iteration $n$, IHT as defined by
\eqref{IHT} will recover an
approximation $\bm{x}^{n}$ satisfying
\begin{equation}\label{IHTGeneralrate}
\|\bm{x} - \bm{x}^{n}\|_2  \leq 2^{-n}\|\bm{x}^s\|_2 + \|\bm{x}-\bm{x}^s\|_2 +
4.34\|\bm{Ax}_{\overline{S}} +
\bm{e}\|_2. 
\end{equation} 
\end{theorem}

To pass from \eqref{IHTGeneralrate} to \eqref{IHTrate}--\eqref{IHTerrorterm},
Blumensath and
Davies used the following energy bound for RIP matrices from \cite{cosamp}:
\begin{prop}\label{CosampRIPprop}
Suppose that $\bm{A}$ verifies the upper inequality
$$
\|\bm{Ax}\|_2 \leq \sqrt{1+\delta_s}\|\bm{x}\|_2, \textup{ when } \|\bm{x}\|_0
\leq s. 
$$ 
Then, for every signal $\bm{x}$,
\begin{equation}\label{CosampRIP}
 \|\bm{Ax}\|_2 \leq \sqrt{1+\delta_s}\left[  \|\bm{x}\|_2 +
\frac{1}{\sqrt{s}}\|\bm{x}\|_1 
\right].
\end{equation}
\end{prop}
Applying \eqref{CosampRIP} to $\bm{Ax}_{\overline{S}}$ in
\eqref{IHTGeneralrate} yields \eqref{IHTrate}--\eqref{IHTerrorterm}. 

The proof of Proposition \ref{CosampRIPprop} boils down to establishing an
inclusion of
polar spaces: $S^\circ \subset K^\circ$. $S^\circ$ is equipped with the
following norm
$$
\|\bm{u}\|_{S^\circ} = \max_{|I|\leq r} \|\bm{u}_I\|_2. 
$$
The proof proceeds by considering any element $\bm{u}$ of the unit ball in
$S^\circ$. We decompose $\bm{u}$ into two
components: $\bm{u}_S$ and $\bm{u}_{\overline{S}}$ where $\bm{u}_S$ represents
the best
$s$-sparse approximation to $\bm{u}$ in the $\ell_2$ norm. As $S$ contains the
$s$ most energetic atoms, this implies that
the set $\overline{S}$ contains atoms whose energy must lie under a certain
threshold: the bound $\|\bm{u}_{\overline{S}}\|_\infty \leq \frac{1}{\sqrt{s}}$
is easily
obtained. In other words, the following decomposition is obtained:
$$
\bm{u} = \bm{u}_S + \bm{u}_{\overline{S}} \in B_2 + \frac{1}{\sqrt{s}}B_\infty,
$$
and the space on the right hand side is exactly the space $K^\circ$. For
further details, consult \cite{cosamp}, in this article we will only be
concerned
with this particular aspect of their proof.

This sort of decomposition does not hold for the weighted case. Consider the
example in which the weight 
vector $\omega$ is such that $\omega(j) = \sqrt{j}$. As mentioned before, with
such a weight vector $\omega$, any $s$ sparse signal cannot have any atom of
index higher than $\ceil{s^{1/2}}$ supported. Therefore, taking the best
$(\omega,s)$-sparse approximation to a signal does not constrain the
$\ell_\infty$ norm of the signal on the complement $\overline{S}$.
As a result of this, Proposition \ref{CosampRIPprop} does not extend to the
weighted case and an alternative method will be needed to bound the energy of
$\bm{Ax}_{\overline{S}}$. Here we see a key difference between unweighted
sparsity and weighted sparsity: significant amounts of energy can be
concentrated in the tail $\bm{x} - H_{\omega,s}(\bm{x})$. More specifically we
see that certain weight vectors can yield the process of
taking the best $(\omega,s)$-sparse approximation to be an operation which is
inherently
local as it may restrict the analysis to lie on a subset of low weight atoms and
the higher weight atoms are completely ignored.

An alternative method of bounding the term $\|\bm{Ax}\|_2$ involves a different
type of decomposition. Suppose $\bm{A}$ satisfies the RIP of order $s$ with RIP
constant $\delta_s$. Let $\{S_i\}_{i = 1}^{p}$ be a partition of $[N]$ into
$s$-sparse blocks: each $S_i$ satisfies: for all $i\neq j, S_i \cap S_j =
\emptyset$ and $\textrm{card}(S_i) \leq s$. We may apply the RIP to each $S_i$
block to obtain the following bound:
\begin{align*} 
\|\bm{Ax}\|_2 = \| \bm{A} (\displaystyle \sum_i \bm{x}_{S_i})  \|_2 &\leq \sum_i
\|\bm{Ax}_{S_i}\|_2
\\&\leq \sqrt{1+\delta_s} \sum_i \|\bm{x}_{S_i}\|_2 
\\&\leq \sqrt{ (1+\delta_s)/s   } \sum_i \|\bm{x}_{S_i}\|_1 = \sqrt{
(1+\delta_s)/s   }\|\bm{x}\|_1 \label{RIPSblocksend}
\end{align*}
This sort of argument in general will not extend to the weighted case. Depending
on the weight vector $\omega$ such a decomposition of an arbitrary signal into a
collection of disjoint
$s$-sparse blocks may not even be possible. 

Therefore, the performance guarantee given in Theorem \ref{IHTguarantee} does
not
directly extend to
the weighted case. The more general guarantee of Theorem
\ref{IHTGeneralGuaranteeTHM}
however, does extend, and the proof is identical to the proof in the unweighted
case. 
\begin{theorem}\label{IHWTguarantee} Let $\omega \in \mathbb{C}^N$ denote a
weight vector with $\omega(i) \geq 1$ for all $i = 1,\dots, N$. 
Let $\bm{y} = \bm{Ax} + \bm{e}$ denote a set of noisy observations where $\bm{x}$ is an
arbitrary vector. If $\bm{A}$ has weighted restricted isometry
property with
$\delta_{\omega,3s} < 1/\sqrt{32}$, then at iteration $n$, IHWT as defined by
\eqref{IHWT} will recover an
approximation $\bm{x}^{n}$ satisfying
\begin{equation}\label{IHWTrategeneric}
\|\bm{x} - \bm{x}^{n}\|_2  \leq 2^{-n}\|\bm{x}^s\|_2 + \|\bm{x}-\bm{x}^s\|_2 +
4.34\|\bm{Ax}_{\overline{S}} +
\bm{e}\|_2. 
\end{equation} 
\end{theorem}
\begin{proof}
We follow the proof presented in \cite{IHT}. By the triangle inequality we have
that:
\begin{equation}\label{IHWToptimalityPass0}
\|\bm{x} - \bm{x}^{n+1}\|_2 \leq \|\bm{x} - \bm{x}^s\|_2 + \|\bm{x}^{n+1} -
\bm{x}^s\|_2. 
\end{equation}
We focus on the term $ \|\bm{x}^{n+1} - \bm{x}^s\|_2$. This term is supported on
$T^{n+1}$ and we may therefore restrict our analysis to this index set. By the
triangle inequality we have:
$$
\|\bm{x}^s - \bm{x}^{n+1}\|_2 \leq \| \bm{x}^s_{T^{n+1}} -
\bm{a}^{n+1}_{T^{n+1}}\|_2  + \|\bm{x}^{n+1}_{T^{n+1}} -
\bm{a}^{n+1}_{T^{n+1}}\|_2 
$$
By definition of the thresholding operator $H_{\omega,s}$, the signal
$\bm{x}^{n+1}$ is the best weighted $s$ sparse approximation to $\bm{a}^{n+1}$. In
particular, $\bm{x}^{n+1}$ is a better weighted $s$ sparse approximation to
$\bm{a}^{n+1}$ than $\bm{x}^s$. We therefore obtain the inequality:
$$
\|\bm{x}^s - \bm{x}^{n+1}\|_2 \leq 2\|\bm{x}^s_{T^{n+1}} -
\bm{a}^{n+1}_{T^{n+1}}\|_2. 
$$
Expanding the term $\bm{a}^{n+1}$:
\begin{align*}
\|\bm{x}^s - \bm{x}^{n+1}\|_2 & \leq 2\|\bm{x}^s_{T^{n+1}} - \bm{x}^n_{T^{n+1}} -
\bm{A}^*_{T^{n+1}}\bm{Ar}^n + \bm{A}^*_{T^{n+1}}\bm{Ax}_{\overline{S}} +
\bm{A}^*_{T^{n+1}}\bm{e}\|_2
\\&\leq 2\|\bm{r}^n_{T^{n+1}} - \bm{A}^*_{T^{n+1}}\bm{Ar}^n\|_2 + 2\|
\bm{A}^*_{T^{n+1}}(\bm{Ax}_{\overline{S}} +\bm{e})\|_2
\\&= 2\|(\bm{I} - \bm{A}^*_{T^{n+1}}\bm{A}_{T^{n+1}}) \bm{r}^n_{T^{n+1}} -
\bm{A}^*_{T^{n+1}}\bm{A}_{T^n \setminus T^{n+1}}\bm{r}^n_{  T^n \setminus
T^{n+1} }\|_2 
\\ & \indent+  2\| \bm{A}^*_{T^{n+1}}(\bm{Ax}_{\overline{S}} +\bm{e})\|_2
\\ & \leq 2\|(\bm{I} - \bm{A}^*_{T^{n+1}}\bm{A}_{T^{n+1}})
\bm{r}^n_{T^{n+1}}\|_2 + 2\| \bm{A}^*_{T^{n+1}}\bm{A}_{T^n \setminus
T^{n+1}}\bm{r}^n_{  T^n \setminus T^{n+1} }\|_2 
\\ & \indent +  2\|\bm{A}^*_{T^{n+1}}(\bm{Ax}_{\overline{S}} +\bm{e})\|_2
\end{align*}
Note that $T^n \setminus T^{n+1}$ is disjoint from $T^{n+1}$ and $\|T^n
\cup T^{n+1}\|_{\omega,0} \leq 3s$. Applying the RIP bounds from \ref{classicRIPbounds}:
\begin{align*}
\|\bm{r}^{n+1}\|_2 &\leq 2\delta_{\omega,2s}\|\bm{r}^{n}_{T^{n+1}}\|_2 +
2\delta_{\omega,3s}\|\bm{r}^n_{T^n \setminus T^{n+1}}\|_2 + 2\sqrt{ 1 +
\delta_{\omega,2s}  } \|\bm{Ax}_{\overline{S}} + \bm{e}\|_2
\\ &\leq 2 \delta_{\omega,3s} \left(  \|\bm{r}^{n}_{T^{n+1}}\|_2 +
\|\bm{r}^n_{T^n \setminus T^{n+1}}\|_2 \right) +  2\sqrt{ 1 + \delta_{\omega,
3s}  } \|\bm{Ax}_{\overline{S}} + \bm{e}\|_2
\\ &\leq \sqrt{8} \delta_{\omega,3s} \|\bm{r}^n\|_2 +  2\sqrt{ 1 +
\delta_{\omega, 3s}  } \|\bm{Ax}_{\overline{S}} + \bm{e}\|_2.
\end{align*}
If we have that $\delta_{\omega, 3s}< \frac{1}{\sqrt{32}}$, then
$$
\|\bm{r}^{n+1}\|_2 \leq 0.5\|\bm{r}^n\|_2 + 2.17 \|\bm{Ax}_{\overline{S}} + e
\|_2. 
$$
Iterating this relationship and using the fact that $\sum_{i = 0}^\infty 2^{-i}
= 2$, we obtain the bound:
\begin{equation}\label{IHWToptimalityPass}
\|\bm{r}^{n}\|_2 < 2^{-n} \|\bm{x}^s\|_2 + 4.34  \|\bm{Ax}_{\overline{S}} +
\bm{e}\|_2.
\end{equation}
Combining \eqref{IHWToptimalityPass} with \eqref{IHWToptimalityPass0} proves
the desired claim. 
\end{proof}

Note that the proof has two main components: the hard thresholding operator
produces $\bm{x}^{n+1}$, a superior sparse approximation to the gradient descent
update $\bm{a}^{n+1}$ than $\bm{x}^s$ and applying the RIP. Moreover, the proof
never requires any details of the projection or even the space we are projecting
onto, unlike the proof of Proposition \ref{CosampRIPprop}, which uses special
properties of
the projection onto the space of unweighted $s$ sparse signals. This is
precisely why the IHWT performance guarantee and its corresponding analysis is
nearly identical to the IHT guarantee from Theorem \ref{IHTGeneralGuaranteeTHM}.

The existence of weights which are known to produce signals with heavy tails is
the main blockage to the extension of some more detailed performance guarantees,
like that of Theorem \ref{IHTGeneralGuaranteeTHM}. In the two cases before, one
notices
that the existence of heavy tailed signals prevented a decomposition of a signal
amenable to further analysis. However, for certain \textit{bounded} weight
parameters, for arbitrary signals $\bm{x}$, one may obtain a modified bound on
$\|\bm{Ax}\|_2$ in terms of \textit{weighted norms}. Indeed we obtain the
following weighted analogue of Proposition \ref{CosampRIPprop}:
\begin{prop}\label{WeightedRIPprop}
Consider a sparsity level $s$ and a weight parameter $\omega$ satisfying $s \geq
2\|\omega\|_\infty^2$. If $\bm{A}$ satisfies the $\omega$-RIP of order $s$ with
RIP constant $\delta_{\omega, s}$, then the following inequality holds for any
arbitrary signal $\bm{x}$:
\begin{equation}
\|\bm{Ax}\|_2 \leq \sqrt{1+\delta_{\omega,s}} \left( \|\bm{x}\|_2 +
\frac{2}{\sqrt{s}} \|\bm{x}\|_{\omega,1} \right)
\end{equation}
\end{prop} 
\begin{proof}
The proof employs the same strategy used in the proof of Theorem 4.5 in
\cite{RW}.
Let $\bm{x} \in \mathbb{C}^N$. We will partition $[N]$ into weighted $s$ sparse
blocks $S_1, \dots, S_{p}$ for some index $p$ with each block satisfying $s -
\|\omega\|_\infty^2 \leq \omega(S_l) \leq s$. Furthermore, we assume that the
blocks $S_i$ are formed according to a \textit{nonincreasing rearrangement} of
$\bm{x}$ with respect to the weights, i.e.
\begin{equation}\label{nonincreasing-rearrangement}
|x_j| \omega_j^{-1} \leq |x_k|\omega_k^{-1} \textrm{ for all } j \in S_{l}
\textrm{ and for all } k \in S_{l-1}, l\geq 2.
\end{equation}
For any $k \in S_l$, set $\alpha_k = (\sum_{j \in S_l} \omega_j^2)^{-1}
\omega_k^2 \leq (s-\|\omega\|_\infty^2)^{-1} \omega_k^2$ by hypothesis. Notice
that $\sum_{k \in S_l} \alpha_k = 1$. For $l \geq 2$ then:
\begin{align}
|x_j| \omega_j^{-1} & \label{convexweightedaverage} \leq \sum_{k \in S_{l-1}}
\alpha_k |x_k| \omega_k^{-1} \textrm{ for any } j \in S_l
\\ &\label{linftynorm}\leq (s - \|\omega\|_\infty^2)^{-1} \sum_{k \in S_{l-1}}
|x_k| \omega_k^{-1} \omega_k^2 
\\&=  (s - \|\omega\|_\infty^2)^{-1} \sum_{k \in
S_{l-1}} |x_k| \omega_k
\\&  = (s - \|\omega\|_\infty^2)^{-1} \|\bm{x}_{S_{l-1}}\|_{\omega,1}. 
\end{align}
where \eqref{convexweightedaverage} holds by nonincreasing rearrangement and
convexity and \eqref{linftynorm} holds by hypothesis. 
Therefore, by the Cauchy-Schwarz inequality, we obtain:
$$
\|\bm{x}_{S_l}\|_2 \leq \frac{\sqrt{s}}{ s - \|\omega\|_\infty^2  }
\|\bm{x}_{S_{l-1}}\|_{\omega,1} \leq
\frac{2}{\sqrt{s}}\|\bm{x}_{S_{l-1}}\|_{\omega,1} \textrm{ for } l \geq 2.
$$

For $\|\bm{Ax}\|_2$, we obtain the following estimate:

\begin{align*}
\|\bm{Ax}\|_2 & \leq \sum_{i =1}^p \|\bm{Ax}_{S_i}\|_2 
\\&\leq \sqrt{1+\delta_{\omega,s}}  \sum_{i =1}^p \|\bm{x}_{S_i}\|_2   
\\&= \sqrt{1+\delta_{\omega,s}}\left(  \|\bm{x}_{S_1}\|_2  +  \sum_{i =2}^p
\|\bm{x}_{S_i}\|_2    \right)
\\&\leq \sqrt{1+\delta_{\omega,s}} \left( \|\bm{x}_{S_1}\|_2 +
\frac{2}{\sqrt{s}}\sum_{i = 1}^{p-1} \|\bm{x}_{S_i}\|_{\omega,1} \right)
\\&\leq \sqrt{1+\delta_{\omega,s}} \left( \|\bm{x}\|_2 + \frac{2}{\sqrt{s}}
\|\bm{x}\|_{\omega,1} \right).
\end{align*}
\end{proof}

Applying \ref{WeightedRIPprop} to $\|\bm{Ax}_{\overline{S}}\|_2$ immediately
yields the following performance bound:
\begin{theorem}\label{IHWTextendedguarantee} For sparsity level $s$, let $\omega
\in \mathbb{C}^N$ denote a
weight vector with $\omega(i) \geq 1$ for all $i = 1,\dots, N$ satisfying $s
\geq 2\|\omega\|_\infty^2$. 
Let $\bm{y} = \bm{Ax} + \bm{e}$ denote a set of noisy observations where $\bm{x}$ is an
arbitrary vector. If $\bm{A}$ has weighted restricted isometry
property with
$\delta_{\omega,3s} < 1/\sqrt{32}$, then at iteration $n$, IHWT as defined by
\eqref{IHWT} will recover an
approximation $\bm{x}^{n}$ satisfying
\begin{equation}\label{IHWTrate}
\|\bm{x} - \bm{x}^{n}\|_2  \leq 2^{-n}\|\bm{x}^s\|_2 + 6 \left(
\|\bm{x}-\bm{x}^s\|_2 + \frac{2}{\sqrt{s}}\|\bm{x}-\bm{x}^s\|_{\omega,1} +
\|\bm{e}\|_2\right). 
\end{equation} 
\end{theorem}
This result bears a striking resemblance to Theorem \ref{IHTguarantee} except
that it
is in terms of the weighted $\ell_1$ norm as opposed to the unweighted $\ell_1$
norm. 

\subsubsection{Performance Guarantees: Contraction}

For an arbitrary, possibly dense signal $\bm{x}$, the performance guarantees
presented above do not guarantee the convergence of IHT/IHWT, but rather they
guarantee that if the sampling matrix $\bm{A} \in \mathbb{C}^{m \times N}$
satisfies the RIP of order $3s$ then the iterates are guaranteed to converge to
a \textit{neighborhood} of the true best $s$-sparse approximation. In
\cite{IHT2}, alternative guarantees are derived under an alternative assumption
on $\bm{A}$, namely that $\|\bm{A}\|_2 < 1$. In particular, we focus on the
guarantee that if $\bm{A}$ satisfies the spectral bound $\|\bm{A}\|_2 < 1$, then
the sequence of IHT iterates $(\bm{x}_n)$ is a contractive sequence. 

Note that if $\bm{A}$ satisfies the RIP of order
$3s$, then by applying $\bm{A}$ to the canonical Euclidean basis vectors
$\{e_i\}_{i =1}^N$ it follows that the $\ell_2$ norm of the columns of $\bm{A}$
must satisfy:
\begin{equation}\label{RIPspectral}
1 - \delta_{3s} \leq \|A_j\|_2 \leq 1 + \delta_{3s}, \textup{ for } j = 1,
\dots, N. 
\end{equation}
On the other hand, the spectral norm of a linear map can equivalently be
interpreted as an operator norm: $\|\bm{A}\|_2 = \sup_{\|\bm{x}\|_2 =
1}\|\bm{Ax}\|_2$. As a consequence:
$$
\|\bm{A}\|_2 = \sup_{\|\bm{x}\|_2 = 1}\|\bm{Ax}\|_2 \geq \max_{  \bm{e}_i, i =
1, \dots, N} \|\bm{A}\bm{e}_i\|_2 = \max_{i = 1,\dots,N} \|A_i\|_2.
$$
Therefore if $\bm{A}$ satisfies the RIP condition, it could be true that
$\max_{i = 1,\dots,N} \|A_i\|_2 > 1$ by \eqref{RIPspectral}. In this manner, the
RIP condition is in general not compatible with the spectral condition
$\|\bm{A}\|_2 < 1$. 

Observe that if the spectral norm of $\bm{A}$ is bounded above by 1, then the
loss function $f(\bm{x}) = \frac{1}{2} \|\bm{y} - \bm{A} \bm{x}\|_2^2$ is
\textit{majorized} by the following surrogate objective function:
\begin{equation} \label{surrogate}
g(\bm{x}, \bm{z}) = \frac{1}{2} \|\bm{y} - \bm{A} \bm{x}\|_2^2 - \|\bm{A}(\bm{x}
- \bm{z})\|_2^2 + \|\bm{x} - \bm{z}\|_2^2. 
\end{equation}
Because $g(\bm{x}, \bm{x}) = f(\bm{x})$, optimizing $g(\bm{x}, \bm{x})$ will
decrease the objective function $f(\bm{x})$. This is known as Lange's
\textit{Majorization Minimization} (MM) Method \cite{lange}. 

Viewing $\bm{z}$ as fixed, we may decouple the coordinates $x_i$:
\begin{equation}\label{decouple}
g(\bm{x}, \bm{z}) \propto \sum_i x_i^2 - 2x_i(z_i + A^*_i \bm{y} - A^*_i
\bm{A}\bm{z}).
\end{equation}
Ignoring the sparsity constraint on $\bm{x}$, minimizing \eqref{decouple} we
obtain the unconstrained minima $\bm{x}^*$ given by:
$$
x_i^* = z_i + A_i^*\bm{y} - A_i^*\bm{A} \bm{z}.
$$
We then have that:
$$
g(\bm{x}^*, \bm{z}) \propto \sum_i {x^*_i}^2 - 2x_i^*(z_i + A_i^*\bm{y} - A_i^*
\bm{A}\bm{z}) = \sum_{i} -{x^*_i}^2.
$$
Therefore the \textit{$s$-sparse constrained minimum} of the majorizing
surrogate $g$ is given by hard thresholding $\bm{x}^*$ by choosing the largest
$s$ components in magnitude. Clearly the above analysis holds for weighted
sparse approximations as well. We therefore conclude that both the IHT and IHWT iterates share the property that the sparsity constrained minimizer of
$g(\bm{x}, \bm{x}^{n})$ is given by $\bm{x} = \bm{x}^{n+1}$. 

The following lemma establishes that IHWT makes progress at each iterate.
\begin{lemma}
Assume that $\|\bm{A}\|_2 < 1$ and let $(\bm{x}^n)$ denote the IHWT iterates
defined by \eqref{IHWT}. Then the sequences $(f(\bm{x}^n))$ and
$(g(\bm{x}^{n+1}, \bm{x}^n))$ are non-increasing. 
\end{lemma}
\begin{proof}
We have the following sequence of inequalities:
\begin{align*}
f(\bm{x}^{n+1}) &\leq f(\bm{x}^{n+1})  + \|\bm{x}^{n+1} - \bm{x}^n\|_2^2 -
\|A(\bm{x}^{n+1} - \bm{x}^n)\|_2^2 
\\&= g(\bm{x}^{n+1}, \bm{x}^n)
\\ &\leq g(\bm{x}^{n}, \bm{x}^n)
\\&= f(\bm{x}^n)
\\&\leq f(\bm{x}^n) + \|\bm{x}^n - \bm{x}^{n-1}\|_2^2 - \|A(\bm{x}^n -
\bm{x}^{n-1})\|_2^2
\\&= g(\bm{x}^n, \bm{x}^{n-1}). 
\end{align*}
\end{proof}

Next we present the following lemma which states that the IHWT iterates contract.
\begin{lemma}\label{IHWTsummation}
If the sensing matrix $\bm{A}$ satisfies $\|A\|_2^2 \leq 1 - c < 1$ for some
positive $c \in (0,1)$, then for the IHWT iterates $(\bm{x}^n)$ the following
limit holds: $\lim_{n \rightarrow \infty} \|\bm{x}^{n+1} - \bm{x}^n\|_2^2 = 0$. 
\end{lemma} 
\begin{proof}
By the spectral bound:
$$
\|\bm{A}(\bm{x}^{n+1} - \bm{x}^n)\|_2^2 \leq (1-c) \|\bm{x}^{n+1} -
\bm{x}^n\|_2^2. 
$$
Rearranging terms
$$
\|\bm{x}^{n+1} - \bm{x}^n\|_2^2 \leq \frac{1}{c} \left[  \|\bm{x}^{n+1} -
\bm{x}^n\|_2^2 -   \|A(\bm{x}^{n+1} - \bm{x}^n)\|_2^2 \right].
$$
We define the sequence of partial sums $(s_n)$ by $s_n = \sum_{i = 0}^n
\|\bm{x}^{i+1} - \bm{x}^i\|_2^2$. Clearly the sequence $(s_n)$ is
monotonically increasing. If we can show that the sequence $(s_n)$ is also
bounded, then $(s_n)$ is a convergent sequence. Let $k$ be any arbitrary
index. We then obtain the following sequence of inequalities:

\begin{align}
s_k = \sum_{i = 0}^k \|\bm{x}^{i+1} - \bm{x}^i\|_2^2 &\leq \frac{1}{c} \sum_{i =
0}^k \left(    \|\bm{x}^{i+1} - \bm{x}^i\|_2^2 - \|A(\bm{x}^{i+1} -
\bm{x}^i)\|_2^2     \right)
\\& =\frac{1}{c} \sum_{i = 0}^k g(\bm{x}^{i+1}, \bm{x}^i) - \frac{1}{2} \|\bm{y}
- \bm{A}\bm{x}^{i+1}\|_2^2
\\&\leq \label{surrogateoptimality} \frac{1}{c} \sum_{i = 0}^k g(\bm{x}^{i},
\bm{x}^i) -f(\bm{x}^{i+1}) 
\\& =\frac{1}{c} \sum_{i = 0}^k f(\bm{x}^i)  -f(\bm{x}^{i+1})  
\\&=\frac{1}{c}( f(\bm{x}^0) - f(\bm{x}^{k+1}))
\\& \leq \frac{1}{c} f(\bm{x}^0)
\end{align}
where \eqref{surrogateoptimality} follows from the next IHWT iterate
$\bm{x}^{i+1}$ being a minimizer of $g(\bm{x}, \bm{x}^i)$. 

Therefore $\{s_n\}$ is a convergent sequence. As the sequence of partial sums
converges, the infinite sum $\sum_{ i = 0}^\infty \|\bm{x}^{i+1} -
\bm{x}^i\|_2^2 < \infty$ and as a result $\lim_{n \rightarrow \infty}
\|\bm{x}^{n+1} - \bm{x}^n\|_2^2 = 0$. 
\end{proof}

\section{Iterative Hard Weighted Thresholding: In Practice}

\subsection{Choosing the weights} Before delving into numerical
experiments, we pause for a moment and focus on the overall setup of performing
signal analysis in practice. Note in this article, we have ignored the
preprocessing required to properly select the weight parameter $\omega$. In
reality, this may require either significant domain knowledge (hand crafted) or
the application of a learning algorithm to rank the atoms and assign weights
(learned). If $N \gg 1$, it is not feasible to expect a human expert to assign
weights
to each of these atoms and instead we may assign weights to blocks of atoms.
While this may be effective, the overall structure of the signals may not be
fully captured in such a model. It is an interesting avenue of research to
explore whether or not there are some machine learning algorithms which could
effectively learn the weights of a class of signals given some training data.
One could assume that the weights $\omega$ are generated from some unknown
smooth function $f$, i.e. $\omega(i) = f(i)$ for $i = 1, \dots, N$ and apply
some nonparametric statistical methods. One could test the quality of the
weights by testing to see if weighted $\ell_1$ minimization with those learned
weights can effectively recover test signals. 

Another related problem is to assume that the signals $\bm{x}$ are being generated from some
parameterized probability distribution $p(\bm{x})$.\footnote{It should be noted
that to optimize the parameters, one typically performs some sort of learning
method on training data to optimize the parameters. One common method is to have
some training data and use the Expectation Maximization (EM) algorithm to tune
the parameters.} While it may make intuitive sense why a weight parameter
$\omega$ which is
monotonically increasing is appropriate for a family of power law decay signals,
the manner in which these $\omega_j$ components should grow
is far from obvious. One may pose the following question: given a signal
pdf $p(\bm{x})$, is there an optimal weight parameter $\omega$? Here, optimal
means that with high probability, sparse signals generated from the pdf
$p(\bm{x})$ are recovered from weighted $\ell_1$ minimization with weighted
parameter $\omega$. If so, how does
one compute it? The works \cite{pdf1,pdf2,pdf3} consider this problem and derive
performance guarantees of weighted $\ell_1$ minimization for the optimal weight
parameter $\omega$. In \cite{pdf1}, exact weights were computed for their
simpler signal model in which there are two blocks of support and weights
$\omega_1$ and $\omega_2$ need to be chosen for each block. In \cite{pdf3} a
more general
signal model is employed and the authors suggest methods for choosing the
weights based 
on $p(\bm{x})$. Aside from some relatively simple cases, it is not explicitly
known how to compute an optimal set of weights given a model signal pdf
$p(\bm{x})$.

\subsection{Approximate Projection} The main consequence of the intractability
of
computing weighted best $s$-sparse approximations is that we cannot run the IHWT
algorithm as each iterate requires a projection
onto $\mathcal{S}_{\omega, s}$.

To reconcile this issue we compute an approximation to $H_{\omega,
s}(\bm{x})$. Let $\widetilde{H}_{\omega, s}(\bm{x})$
denote a modified projection operator which sorts the weighted signal
$\omega^{-1} \circ\bm{x}$ \footnote{$A \circ B$ denotes the Hadamard product of
$A$ and $B$.} and thresholds it with respect to the weight $\omega$. Consult
\cite{RW} for properties of this weighted thresholding operator.

In what sense is $\widetilde{H}_{\omega, s}$ an approximation to $H_{\omega,
s}$? We present the following example to build some intuition. Let $N = 100$ and
let $\omega$ be given by $\omega(j) = \sqrt{j}$ for $j = 1, \dots, 100$.
Consider the signal $\bm{x}$ where $x(1) = 10$ and $x(100) = 99$ and equal to 0
otherwise. Then $\omega^{-1} \circ \bm{x} = [10, 0, \cdots, 0, 9.9]$. Sorting
and
thresholding we obtain that $ \widetilde{\bm{x}} = \widetilde{H}_{\omega,
100}(\bm{x}) = [10, 0, \cdots, 0]$. Clearly the best weighted 100 sparse
approximation is given by $\bm{x}^* = [0, \cdots, 0, 99]$. In this case, our
projection operator $\widetilde{H}_{\omega,s}$ did not compute a very good
approximation. However, one can claim that the signal $\bm{x}$ is a
\textit{mis-match} for our weight parameter $\omega$. For signals which ``match"
the weights more closely, $\widetilde{H}_{\omega, s}$ does a better job of
recovering the output of the true projection $H_{\omega, s}$. For example, if
$\bm{x}$ was chosen to be a 
monotonically decreasing signal, this would match the weight $\omega$ and in
this case our surrogate $\widetilde{H}_{\omega, s}$ will compute accurate
projections. 

\subsection{Experiments}

For the remainder of this section, we will be interested in either the approximation or exact recovery of power law distributed signals. To randomly generate power law
signals, we randomly choose integers $a,b$ and formed the power function $f(x) =
\frac{a}{x^b}$ and defined our signal $\bm{x}$ by $\bm{x}(i) = f(i)$ for $i = 1,
\dots, N$. 

We chose our weight parameter $\omega$ as follows: the first $s$-block of coordinates we are relatively certain should be included in our support set as we are dealing with power law signals, thus we set $\omega(1:s) = 1$. For the second $s$-block of coordinates we are more uncertain about their inclusion in the signals support set and thus we set $\omega(s+1:2s) = 3$ and we set the tail $\omega(2s+1:N) = 10$ for similar reasons. Note that these are still relatively mild weights given the power law prior we have assumed. We further note that given these weights the best $(\omega,s)$-sparse approximation is going to be the actual best $s$-sparse approximation for $s$-sparse power law signals.

In the following set of experiments we will test the performance of IHWT for
computing $(\omega,s)$-sparse approximations of dense power law decaying
signals. For arbitrary dense signals $\bm{x}$, it requires a combinatorial
search to compute the best $(\omega,s)$-sparse approximation. However, for power
law decay signals, the best $(\omega,s)$-sparse approximation is simple to
compute as it can be performed by choosing the minimal $k$ index such that $\sum_{i =
1}^k i \leq s$. We note that while the approximate projection operator $\widetilde{H}_{\omega, s}$ will indeed compute the true $(\omega,s)$-sparse approximation of a power law distributed signal, our gradient descent updates $\bm{x}^{n+1} + \bm{A}^*(\bm{y} - \bm{A} \bm{x}^n)$ are a priori not going to be power law distributed signals. Therefore in these experiments, we are not only testing the performance of IHWT, but also of this surrogate projection operator $\widetilde{H}_{\omega, s}$. The noisy measurements were $\bm{y} = \bm{Ax} + \bm{e}$ where
$\bm{e}$ is a Gaussian noise vector. To test the quality of our noisy sparse
approximation, we computed the 
normalized error $\frac{ \|\bm{x}^s-\bm{x}_{\mathrm{approx}}\|_2  }{\|\bm{e}\|_2}$, where
$\bm{x}^s$ is the true best $s$-sparse approximation and $\bm{x}_{\mathrm{approx}}$ is the approximation output by our algorithm.

 \begin{figure}
\centering
 \includegraphics[scale = .8]{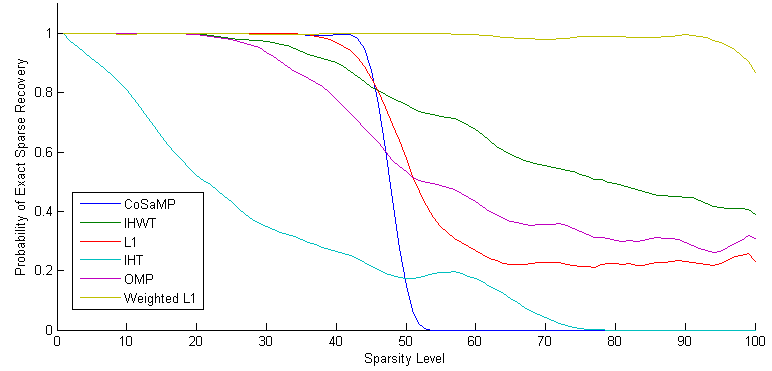}
 \caption{Exact Recovery of Randomly Generated Variable $s$-sparse Power Law Signals using $m  = 128$ measurements. Results are averaged over 200 trials. Best viewed in color. }
 \label{fig:ExactSparseRecoveryVariableMeasurements}
\end{figure}

In Figure \ref{fig:ExactSparseRecoveryVariableMeasurements} we present the performance of IHWT, CoSaMP \cite{cosamp}, IHT \cite{IHT}, OMP \cite{OMP}, $\ell_1$ minimization and weighted $\ell_1$ minimization for the task of exact sparse recovery using $m = 128$ measurements. In particular we randomly generated $\bm{A} \in \mathbb{R}^{128 \times 256}$ Gaussian sensing matrices, $s$-sparse power law signals $\bm{x}_s$ and we have the noise-free measurements $\bm{y} = \bm{A} \bm{x}_s$. We consider a signal to be exactly recovered if the signal approximation and the true underlying signal agree to four decimal places, i.e. $\|\mathrm{x_{approx}} - x_s\|_2 \leq 10^{-4}$. We averaged over 200 trials. Observe that while IHWT does not exactly recover the sparser power law signals as well as CoSaMP or $\ell_1$ minimization, its recovery performance degrades much more gracefully as the sparsity level increases.

 \begin{figure}
\centering
 \includegraphics[scale = 1]{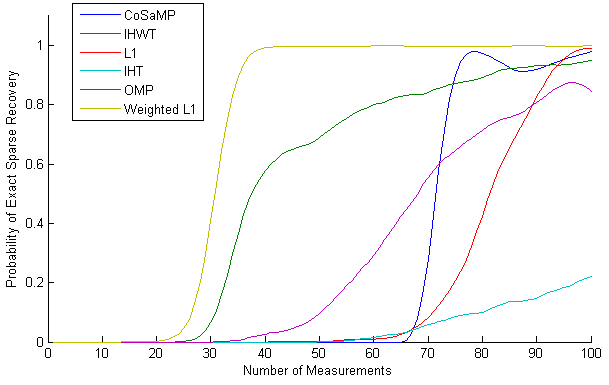}
 \caption{Exact Recovery of a fixed sparse $s = 25$ power law distributed signal using a variable number of measurements. Results are averaged over 200 trials. Best viewed in color. }
 \label{fig:ExactFixedSparsity}
\end{figure}

In Figure \ref{fig:ExactFixedSparsity} we now keep the sparsity level fixed at $s = 25$ and we allow the number of measurements $m$ to vary from 1 to 100. We averaged over 200 runs and we present the probability of exact recovery. Observe the superior performance of IHWT over the other classical greedy sparse approximation algorithms in the undersampling $ m = O(s)$ regime.

In the next set of experiments, we tested the noisy sparse recovery performance of IHWT and we compared it again three standard sparse approximation algorithms: CoSaMP, IHT and OMP. We have a fixed number of measurements $m = 128$ and we randomly generated $\bm{A} \in \mathbb{R}^{128 \times 256}$ Gaussian sensing matrices and we have noisy samples $\bm{y} = \bm{A} \bm{x} + \bm{e}$. Note now that $\bm{x}$ is no longer an $s$-sparse power law distributed signal but rather it is a \textit{dense} power law distributed signal. In Figure \ref{fig:NoisyFixedMeasurementsMean} we present the $\log$ normalized error and log of the standard deviation averaged over 200 trials and in \ref{fig:NoisyFixedMeasurementsStd} we present the $\log$ of the standard deviation of the 200 trials. In Figures \ref{fig:NoisyFixedMeasurementsMean} and \ref{fig:NoisyFixedMeasurementsStd} we see the clear performance advantage of IHWT over other greedy algorithms for the task of fixed sparse approximation of power law distributed signals using a fixed number of measurements.

 \begin{figure}
\centering
 \includegraphics[scale = 1]{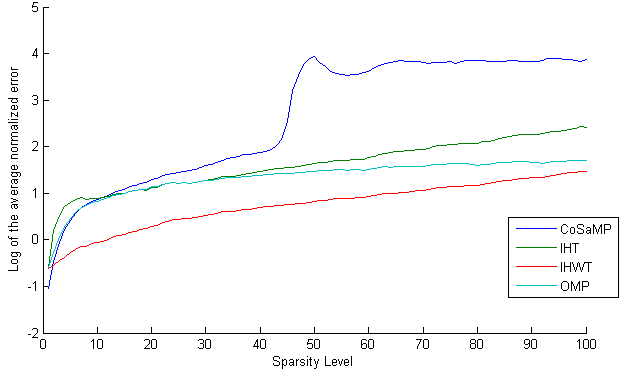}
 \caption{The $\log$ normalized error averaged over 200 trials of noisy $s$-sparse approximation of dense Power Law Signals using $m  = 128$ measurements. Best viewed in color. }
 \label{fig:NoisyFixedMeasurementsMean}
\end{figure}

 \begin{figure}
\centering
 \includegraphics[scale = 1]{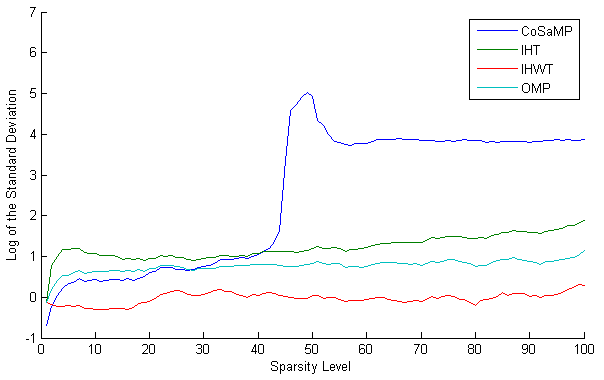}
 \caption{The $\log$ standard deviation of 200 trials of noisy $s$-sparse approximation of dense Power Law Signals using $m  = 128$ measurements. Best viewed in color. }
 \label{fig:NoisyFixedMeasurementsStd}
\end{figure}

In our final set of experiments, we tested how well we could approximate the best $s  = 25$ sparse approximation of a dense power law signal $\bm{x}$ given a set of noisy measurements $\bm{y} = \bm{A} \bm{x} + \bm{e}$ using a variable number of measurements $m = 1, \dots, 100$. In Figures \ref{fig:NoisyFixedSparseMean} and \ref{fig:NoisyFixedSparseStd} we again see the improved performance of IHWT over other standard greedy sparse approximation algorithms. 

 \begin{figure}
\centering
 \includegraphics[scale = 1]{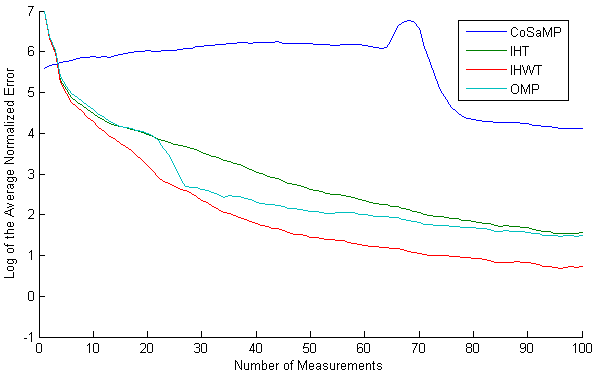}
 \caption{The $\log$ normalized error averaged over 200 trials of noisy $s$-sparse approximation of dense Power Law Signals using a variable number of measurements. Best viewed in color. }
 \label{fig:NoisyFixedSparseMean}
\end{figure}

 \begin{figure}
\centering
 \includegraphics[scale = 1]{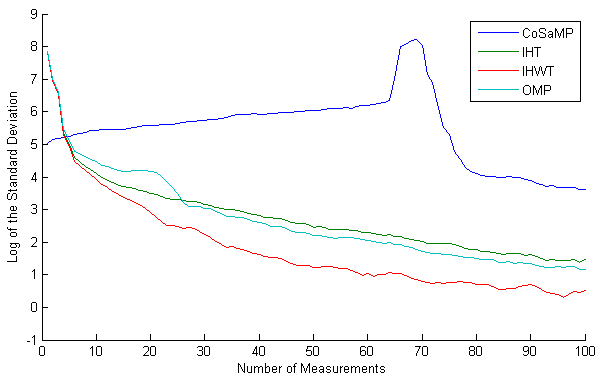}
 \caption{The $\log$ standard deviation of 200 trials of noisy $s$-sparse approximation of dense Power Law Signals using a variable number of measurements. Best viewed in color. }
 \label{fig:NoisyFixedSparseStd}
\end{figure}


\section{Conclusion and Future Directions} 
We have presented the IHWT algorithm which is a weighted extension of the IHT algorithm using the weighted
sparsity technology developed in \cite{RW}. We established theoretical guarantees
of IHWT which are weighted analogues of their unweighted counterparts. While not
all of the guarantees presented in \cite{IHT, IHT2, NIHT} are able to be
extended, in certain cases like Prop \ref{WeightedRIPprop} and Theorem
\ref{IHWTextendedguarantee} we were able to extend the results using the
additional hypothesis that the weight parameter $\omega$ satisfies
$\|\omega\|_\infty \leq O(s)$ for a given weighted sparsity $s$. This condition
allowed us to control the tail $\bm{x} - \bm{x}^s$ and instead of obtaining
$\ell_p$ error bounds, we obtained the analogous error bounds in the weighted
$\ell_p$ norms. Empirically to test the performance of IHWT, we implemented a
tractable surrogate for
the projection onto the space of weighted signals. The numerical experiments also show that the normalized
version of IHWT has superior performance to their
unnormalized counterparts. 

We pose the following open problems:
\begin{enumerate}[1.]
\item Can the results from \cite{IHT2}, which guarantee the convergence of IHT to a local minimizer be extended to IHWT? Here we only extended the guarantee that the IHWT sequence of iterates $(\bm{x}^n)$ is a contractive sequence.
\item Can we learn the weight parameter $\omega$ given some training data $\{\bm{x}^i\}$ where each $\bm{x}^i$ is a known signal?
\item Here we simply used the weights from the weighted $\ell_1$ minimization
problem to be our sparsity weights. However, is there a more optimal choice of
weights to reduce the performance gap between IHWT and weighted $\ell_1$
minimization for the task of exact sparse recovery? 
\end{enumerate}

\section*{Acknowledgements}
The author gratefully acknowledges support from Rachel Ward's Air Force Young
Investigator Award and overall guidance from Rachel Ward throughout the entire
process.

\bibliography{IHWT}

\end{document}